\documentclass[twoside,11pt]{article}
\usepackage{amsfonts,amsmath,amsthm,amssymb,graphicx}
\usepackage[boxruled,vlined,longend]{algorithm2e}
\usepackage{tikz}

\usetikzlibrary{arrows.meta}
\graphicspath{{./images/}{./}}
\newtheorem{proposition}{Proposition}
\newtheorem{remark}{Remark}
\newtheorem{lemma}{Lemma}
\newtheorem{theorem}{Theorem}

\begin{document}

\title{Elementwise Adaptive Random Search for High-Dimensional Constrained Continuous Global Optimization}
\author{Pierre Bertrand$^1$, Michel Broniatowski$^2$, Wolfgang Stummer$^3$ \\
$^{1}$ Aix Marseille Univ, CNRS, AMSE, Marseille, France\thanks{The project leading to this publication has received funding from the French government under the ``France 2030'' investment plan managed by the French National Research Agency (reference: ANR-17-EURE-0020) and from Excellence Initiative of Aix-Marseille University - A*MIDEX.};\\
$^{2}$ LPSM, Sorbonne University, Paris, France;\\
$^{3}$ Department of Mathematics, University of Erlangen-N\"{u}rnberg (FAU), Erlangen, Germany}
\maketitle

\begin{abstract}

We propose a new stochastic method for high-dimensional constrained continuous global optimization without convexity assumptions. The algorithm combines global exploration and local refinement through coordinate-wise adaptive perturbations, whose scale is proportional to the current iterate.

We show that global exploration reaches any feasible region with positive probability, while local search benefits from strong concentration properties, enabling efficient decline near optima. Their combination yields convergence to a global minimizer under mild regularity conditions.

Applications to LASSO sparsity improvement as well as similar BPDN (Basis Pursuit DeNoising) problems, and to neural network pruning, illustrate the practical effectiveness of the approach.

\end{abstract}




\section{Introduction}

High-dimensional constrained continuous optimization remains challenging in the absence of convexity, a situation frequently encountered in modern machine learning tasks such as sparse modeling and neural-network compression. For combinatorial problems, Stochastic Local Search (SLS) methods provide a flexible framework by iteratively exploring the search space through random perturbations of a current solution \cite{ruiz2006stochastic}. Classical combinatorial/continuous approaches include simulated annealing \cite{kirkpatrick1983optimization, corana1987minimizing, dekkers1991global}, iterated local search \cite{lourencco2018iterated, kramer2010iterated, liao2013benchmark}, and variable neighborhood search \cite{hansen1997variable, hansen2019variable}. While effective in moderate dimensions, these methods typically rely on isotropic perturbations, which become inadequate in high dimension where coordinates may require heterogeneous scales; this degradation of global random search in large dimensions is documented and quantified in e.g.  \cite{pepelyshev2018performance, noonan2024improving}.

In this paper, we propose a stochastic optimization framework based on \emph{coordinate-wise adaptive scaling}, where the magnitude of perturbations is proportional to the current iterate. This induces an anisotropic exploration that naturally adapts to the geometry of high-dimensional problems. Our hybrid method combines two complementary regimes: a \emph{Global Search} (GS) phase, which explores the feasible set broadly, and a \emph{Concentrated Search} (CS) phase, which performs local refinement around the best-so-far candidate. Their alternation enables both global coverage and efficient local improvement.

We consider the constrained minimization problem
\[
\min_{s \in \mathcal{S}} f(s),
\]
where $f$ is continuous on a high-dimensional set $\mathcal{S} \subset \mathbb{R}^K$, without convexity assumptions. We assume that $\mathcal{S}$ satisfies $\mathrm{Cl}(\mathrm{Int}(\mathcal{S})) = \mathrm{Cl}(\mathcal{S})$ and consists of finitely many connected components, a standard condition in stochastic optimization \cite{broniatowski2022precise, broniatowski2024precise}.

From a theoretical perspective, we establish two complementary properties. The GS regime reaches any target region of $\mathcal{S}$ with positive probability, while the CS regime benefits from strong high-dimensional concentration, enabling effective local descent under mild regularity assumptions. Their combination yields convergence to a global minimizer $s^*$ with probability one. Our approach is related to recent large-deviation-based methods \cite{broniatowski2022precise, broniatowski2024precise}, but differs by explicitly combining global and local search mechanisms within a unified stochastic scheme.

From a practical perspective, the method is simple, memory-efficient, and requires minimal tuning. We illustrate its effectiveness on two applications: improving sparsity of LASSO/BPDN solutions under performance constraints, and neural network pruning with accuracy guarantees.

\paragraph{Contributions.}
\begin{itemize}
\item A scale-adaptive stochastic search method based on coordinate-wise perturbations;
\item A hybrid GS/CS strategy combining global exploration and local refinement;
\item Probabilistic reachability of any feasible region under mild assumptions;
\item Exploitation of high-dimensional concentration for efficient local search;
\item Empirical validation on constrained non-convex problems in sparse regression and deep learning.
\end{itemize}

\paragraph{Related work.}
Our method is related to metaheuristics such as simulated annealing \cite{kirkpatrick1983optimization, corana1987minimizing, dekkers1991global} and iterated local search \cite{lourencco2018iterated, kramer2010iterated, liao2013benchmark}, here adapted to the continuous setting.
The GS/CS dichotomy also echoes the classical distinction between pure random search and adaptive search in the stochastic global optimization literature \cite{zabinsky1992pure, zabinsky2023hesitant}, of which Hit-and-Run \cite{zabinsky1993improving} is a prominent algorithmic representative. Note, however, that direction-based samplers of this type rely on isotropic direction generation: while feasibility is preserved by sampling along the feasible segment, an isotropic direction does not adapt to heterogeneous coordinate scales, so that the probability of generating improving points degrades in very high dimension — the scale-related part of the limitation we document for simulated annealing (cf.~Sections~\ref{sec:lasso} and~\ref{sec:pruning}).
Unlike isotropic perturbation schemes, our method introduces a coordinate-wise adaptive scaling tailored to high-dimensional heterogeneous structures, while remaining computationally lightweight.

\section{Presentation of the generation scheme\label{sec:generation_scheme}}

The starting point is $c_0=s_0\in \mathcal{S}$. At each iteration $m$, we maintain two vectors in $\mathcal{S}$. The first
one is $s_m$, the current \emph{best-so-far} candidate, i.e., the one with
the lowest objective value encountered up to step $m$. The second one is $c_m
$, a generation center from which new random candidates are sampled.
Although both vectors belong to $\mathcal{S}$, the center $c_m$ is not
required to be close to a minimizer and may be far from the current best
value.

At iteration $m$, we denote by $c_m\in\mathcal{S}$ the \emph{current center}. Firstly a point $a_m$ is generated randomly around this center,
componentwise. More precisely, for each $k\in\{1,\ldots,K\}$, independently,
\begin{equation*}
a_{m}^{k} \sim \mathcal{N}(c_{m}^{k},(\sigma _{m}^{k})^{2}) \, ;
\end{equation*}
where $\sigma_m^k$ is a free scale parameter, whose specification depends on
the regime (GS or CS) and is provided in the corresponding sections below.

The raw draw $a_m$ is then \emph{controlled} to avoid overly large moves
away from $c_m$. For this, we define $\alpha_m$ through the following deterministic
rule: 
\begin{equation}
\alpha_{m}^{k} := 
\begin{cases}
a_{m}^{k}, & \text{if }|a_{m}^{k}-c_{m}^{k}|\leq C\ast |c_{m}^{k}| \, , \\ 
c_{m}^{k}, & \text{otherwise,}%
\end{cases}
\label{fo:alphamk}
\end{equation}
where $C>2$ is a fixed constant (cf. Proposition~\ref{prop:parcours}, $C=2.1$
in our experiments).

In other words, $a_m$ proposes a Gaussian move around the center $c_m$,
while $\alpha_m$ keeps only coordinates whose deviation from $c_m$ remains
controlled; otherwise, the corresponding coordinate is reset to its central
value $c_m^k$.

\section{Pure GS regime}

\subsection{Presentation}

In the GS regime, the objective is to explore the search space $\mathcal{S}$
as broadly as possible.
Following the generation mechanism introduced in Section~\ref%
{sec:generation_scheme}, at step $m\to m+1$ we start from $s_m\in\mathcal{S}$
and $c_m\in\mathcal{S}$, and for each component $k\in\{1,\ldots,K\}$, independently, we
sample 
\begin{equation*}
a_{m}^{k} \sim \mathcal{N}(c_{m}^{k},(\sigma _{m}^{k})^{2}) \, ;
\end{equation*}
in GS we specify the free scale parameter by 
\begin{equation*}
\sigma_m^k:=|c_m^k|.
\end{equation*}
This coordinate-wise choice makes the exploration non-isotropic and naturally adaptive to heterogeneous scales across coordinates, while avoiding the introduction of an additional tuning parameter.
This is deterministically transformed into $\alpha_m$ via the control
rule~\eqref{fo:alphamk}, with $C=2.1$ in our experiments.

In the current GS regime, the update of the control-variable value $c_{m}$
is constructed by 
\begin{equation*}
c_{m+1} := 
\begin{cases}
\alpha_{m}, & \text{if }\alpha_m \in \mathcal{S} \, , \\ 
c_{m}, & \text{otherwise;}%
\end{cases}%
\end{equation*}
moreover, the best-so-far (incumbent) candidate solution $s_{m} \in \mathcal{%
S}$ is updated according to 
\begin{equation*}
s_{m+1} := 
\begin{cases}
\alpha_{m}, & \text{if }\alpha_m \in \mathcal{S} \text{ and }f(\alpha_m) <
f(s_m) \, , \\ 
s_{m}, & \text{otherwise.}%
\end{cases}%
\end{equation*}
Clearly, by construction we have $s_{m+1} \in \mathcal{S}$ and $c_{m+1} \in 
\mathcal{S}$.

\subsection{Properties}

\begin{proposition}[The Pure GS Algorithm Explores All of $\mathcal{S}$]
\label{prop:parcours} Assume that there exists a radius $r>0$ such that $B_r(0)\cap \mathcal{S}=\emptyset$ (in particular $0\notin\mathcal{S}$). Let $\beta$ be a point in $\mathcal{S}$. For
any initial center $c_0 \in \mathcal{S}$ and any $\eta > 0$ small enough to
ensure $B_{\eta}(\beta) \subset \mathcal{S}$, there exists $m_0 \in \mathbb{N%
}$ such that $\mathbb{P} (c_{m_0+1} \in B_{\eta}(\beta)) > 0$.
\end{proposition}

The proof is given in Appendix~\ref{app:proof_prop_parcours}.

The above argument is path-dependent, and no uniform lower bound is assumed
for transitions from an arbitrary center $c_0^{\prime }$ to a target $\beta$; this motivates restarting GS from $c_0$ after blocks of $m_0$ steps.
Moreover, by~\eqref{eq:m0_bound} below, if admissible paths in $\mathcal{S}$ are
uniformly bounded in length, then $m_0$ admits a known uniform upper bound; otherwise, one may relaunch GS with increasing block lengths.

\begin{remark}[Several path-connected components]
If the feasible space $\mathcal{S}$ is split into several path-connected components, one
may generalize Proposition~\ref{prop:parcours} by introducing several
initial centers $c_0^{(1)},\ldots,c_0^{(J)}$ and alternating/restarting GS
from these centers. This extends the reachability statement componentwise
and can improve exploration across disconnected regions.
\end{remark}

\section{Pure CS regime}

\subsection{Presentation}

Recall that at the
beginning of iteration $m\rightarrow m+1$, we have $s_{m}\in \mathcal{S}$
and $c_{m}\in \mathcal{S}$.

In the CS regime, the objective is to explore a small neighborhood of the
current best-so-far candidate $s_m$ while keeping a strong concentration of
the random proposals. This type of concentrated random construction already
appears in~\cite{broniatowski2022precise, broniatowski2024precise}. Following the generation mechanism
introduced in Section~\ref{sec:generation_scheme}, for each component $
k\in\{1,\ldots,K\}$ we sample 
\begin{equation*}
a_{m,n}^{k} \sim \mathcal{N}(s_{m}^{k},(\sigma _{m}^{k})^{2}) \, ;
\end{equation*}
and in CS we define the scale parameter by
\begin{equation*}
\sigma_m^k:=\frac{|s_m^k|}{\sqrt{nK}},
\end{equation*}
where $n$ is a tuning parameter. In high-dimension, the scale parameter is huge by default since it's a multiple of $K$. As in GS, this coordinate-wise scaling is non-isotropic and adaptive to heterogeneous coordinates; the factor $nK$ controls concentration without introducing per-coordinate tuning parameters. For consistency with GS notation, we omit the index $n$ and simply write $\alpha_m$.

Although this parameter choice may seem ad hoc, it is well-calibrated: by scaling each standard deviation with the coordinate magnitude and the dimension $K$, it ensures strong concentration of proposals around $s_m$ in high dimension, as formalised in Lemma~\ref{lem:voisinage} below.

Skipping the control filter in the current CS regime, we directly set $%
a_m:=\alpha_{m}$. Then the update of the control-variable value $c_{m}$ is
constructed by 
\begin{equation*}
c_{m+1} := c_{m},
\end{equation*}
and the best-so-far (incumbent) candidate solution $s_{m} \in \mathcal{S}$
is updated according to 
\begin{equation*}
s_{m+1} := 
\begin{cases}
\alpha_{m}, & \text{if }\alpha_m \in \mathcal{S} \text{ and }f(\alpha_m) <
f(s_m) \, , \\ 
s_{m}, & \text{otherwise.}%
\end{cases}%
\end{equation*}
Clearly, by construction we have $s_{m+1} \in \mathcal{S}$ and $c_{m+1} \in 
\mathcal{S}$.

\subsection{Properties}

\begin{proposition}[The Pure CS Algorithm converges to a global minimizer]
\label{prop:cs_converges} Let $s^*$ be a global minimizer. Assume there exists $\epsilon >0$ such that, on $B_{\epsilon}(s^*)\subset\mathcal{S}$, $f$ decreases with the distance to $s^*
$. Then CS with start $s_0\in B_{\epsilon}(s^*)$ reaches an arbitrarily small neighborhood of $s^*$.
\end{proposition}

For clarity we simplify several notations: the standard-deviation
vector $\sigma _{m}=(\sigma _{m}^{1},\ldots ,\sigma _{m}^{K})$ is quoted as $%
\sigma$; and we denote by $\alpha_{m,n}(\sigma)$ the CS random draw (where $m
$ is the iteration index and $n$ is the concentration parameter; in this
subsection we keep the pair $(m,n)$ explicit). Thus, $\alpha_{m,n}(\sigma)$
is a Gaussian random vector in $\mathbb{R}^{K}$, with independent
components, centered at $s_m = \left(s_m^1,\dots,s_m^K\right)$, and with
covariance matrix $\sigma I_{\mathbb{R}^K}\sigma^T$.

Denote by $V_{\bar{\sigma}} \in \mathbb{R}^{K}$ the vector for which all
components are equal to $\bar{\sigma}=\max_{k}\sigma ^{k}$, and let
\begin{equation*}
S_{n}^{2}:=\frac{(\bar{\sigma})^{2}}{nK}.
\end{equation*}

\begin{lemma}[The Pure CS regime explores the neighborhood of its center]
\label{lem:voisinage} The random vector $\alpha_{m,n}(\sigma)$ concentrates
in an annulus around $s_m$, with probability lower bounded by $\eta\geq 0$
independently of $n$.

Formally, there exists a constant $\eta_K \in (0,1)$, such that for any $\eta \in \left( 0,\eta_K\right)$ there exists $\theta \in \left( 0,1\right) $
such that for all $n$ (multiple of $K$) it holds 
\begin{equation*}
\mathbb{P}\left( ||\alpha_{m,n}(\sigma)-s_m||_{2}\in D_{\eta }\right) >\eta
\end{equation*}
where $D_{\eta}:=\left( D_{\eta }^{-},D_{\eta }^{+}\right) $ with $D_{\eta
}^{-}=(1-\theta )S_{n}\sqrt{K}$ and $D_{\eta }^{+}=(1+\theta )S_{n}\sqrt{K}$.
\end{lemma}

The proof of Lemma~\ref{lem:voisinage} is given in Appendix~\ref{app:proof_lem_voisinage}.

\begin{proof}[Proof of Proposition~\ref{prop:cs_converges}]
Start from $s_0\in B_{\epsilon}(s^*)$. By Lemma~\ref{lem:voisinage}, each CS step has a nonzero probability of generating a candidate in a controlled annulus around its center.
Under the monotonicity assumption of the proposition (inside $B_{\epsilon}(s^*)$, lower objective values correspond to smaller distance to $s^*$), any accepted improvement decreases the distance to $s^*$ by a strictly positive amount with positive probability.

Therefore, by iteration in $m$, the distance process can be reduced below any prescribed radius with strictly positive probability; equivalently, CS reaches arbitrarily small neighborhoods of $s^*$.
\end{proof}

We now prove that Pure CS, although providing converging sequences towards $s^*$%
, is trapped in local minima of $f$, which is the core reason to introduce
GS; indeed it holds

\begin{proposition}
\label{prop:ring_trap} Assume there exists an annulus 
\begin{equation*}
D_{\alpha}(s_0):=B_{2\alpha}(s_0)\setminus B_{\alpha}(s_0) 
\end{equation*}
such that $\min_{x\in D_{\alpha}(s_0)} f(x) > f(s_0)$, and $s^*\notin
B_{\alpha}(s_0)$. Then, for $n$ large enough, the CS regime cannot reach $s^*$ with probability arbitrarily close to $1$.
\end{proposition}

\begin{figure}[ht]
\centering
\begin{tikzpicture}[scale=0.9,>=Stealth]
  \draw[thick] (-2.9,-1.8)
    .. controls (-3.5,-0.9) and (-3.2,0.9) .. (-2.3,1.9)
    .. controls (-1.7,3.0) and (-0.3,3.3) .. (0.8,2.8)
    .. controls (1.4,3.4) and (2.9,3.1) .. (3.7,2.2)
    .. controls (4.8,1.4) and (4.9,-0.2) .. (4.1,-1.1)
    .. controls (3.8,-2.1) and (2.9,-2.9) .. (1.9,-2.6)
    .. controls (0.9,-3.1) and (-0.5,-3.0) .. (-1.3,-2.4)
    .. controls (-2.3,-2.8) and (-3.2,-2.4) .. cycle;
  \node[fill=white,inner sep=1pt] at (3.9,-2.3) {$\mathcal S$};

  \coordinate (S) at (0,0);
  \draw[thick] (S) circle (0.75);
  \draw[thick] (S) circle (1.45);
  \fill (S) circle (1.8pt);
  \node[below right,fill=white,inner sep=1pt] at (0.08,-0.08) {$s_0$};

  \fill (0.60,0.34) circle (1.6pt);
  \node[above right,fill=white,inner sep=1.5pt] at (0.74,0.44) {$s_m$};

  \draw[thick] (0,0) -- (0.14,0.06) -- (0.24,0.20) -- (0.37,0.15) -- (0.49,0.30) -- (0.60,0.34);

  \draw[->] (0,0) -- (-0.75,0);
  \node[above,fill=white,inner sep=1pt] at (-0.38,0) {$\alpha$};

  \draw[->] (0,0) -- (-1.03,-1.03);
  \node[above,fill=white,inner sep=1.5pt] at (-0.72,-0.52) {$2\alpha$};

  \fill (3.5,1.0) circle (1.6pt);
  \node[right,fill=white,inner sep=1.5pt] at (3.66,1.10) {$s^*$};
\end{tikzpicture}
\caption{Annulus geometry in Proposition~\protect\ref{prop:ring_trap}: inner
ball $B_{\protect\alpha}(s_0)$, outer ball $B_{2\protect\alpha}(s_0)$, and
annulus thickness $\protect\alpha$.}
\label{fig:ring_trap_annulus}
\end{figure}

\begin{proof}[Proof of Proposition~\ref{prop:ring_trap}]
In CS, we have $a_m^k\sim\mathcal N\!\left(s_m^k,\,\left(\frac{|s_m^k|}{\sqrt{nK}}\right)^2\right)$.
Hence the typical jump size is of order $1/(nK)$: for every $\rho>0$,
\[
\mathbb P\!\left(\|a_m-s_m\|_2\le \rho\right)\xrightarrow[n\to\infty]{}1,
\]
which advocates for choosing $n$ large. Fix accordingly $\rho<\alpha $ and choose $n$ large enough so that this event has probability arbitrarily close to $1$ at each step. Then, starting from $s_{m}\in B_{\alpha }(s_{0})$, one
cannot jump across the annulus in one CS step: the iterate remains in $B_{2\alpha }(s_{0})$ with probability close to $1$.

Moreover, every point $x$ in the annulus $D_{\alpha }(s_{0})$ is unfavorable for CS, since $f(x)>f(s_{0})$ and $f(s_{0})\leq f(s_{m})$; therefore the CS acceptance rule rejects such points $a_{m}$ as new incumbents. By induction, CS incumbents remain in $B_{\alpha }(s_{0})$ (at most on its boundary) and do not cross the annulus. Because $c^{\ast }\notin B_{\alpha }(s_{0})$, it is not reachable by CS.
\end{proof}

\section{Combination strategy\label{sec:algo}}

\subsection{Presentation\label{ssec:alternate_idea}}

We initialize the procedure with $c_0\in\mathcal{S}$ and $s_0:=c_0$. Hence,
at the initial step, $s_0$ is the current best value and $c_0$ is the
generation center.

The algorithm then alternates the two regimes every iteration: GS, then CS,
then GS again, and so on. In this alternating scheme, GS evolves from its
own center variable, while CS starts from the current best-so-far candidate $%
s$ and performs local refinement.

By the GS exploration property (Proposition~\ref{prop:parcours}), starting
from $c_0$ there exists a finite horizon $m_0$ such that any target region
in $\mathcal{S}$ can be reached with strictly positive probability. To
preserve this global exploration capability during alternation, we reset the
GS center to $c_0$ after block of $m_0$ GS iterations. This periodic
reinitialization prevents long-run drift of the GS center and maintains
repeated large-scale coverage of the search space, while CS keeps improving
the incumbent solution locally.

\begin{figure}[ht]
\centering
\begin{tikzpicture}[scale=0.9,>=Stealth]
  \draw[thick] (-2.9,-1.8)
    .. controls (-3.5,-0.9) and (-3.2,0.9) .. (-2.3,1.9)
    .. controls (-1.7,3.0) and (-0.3,3.3) .. (0.8,2.8)
    .. controls (1.4,3.4) and (2.9,3.1) .. (3.7,2.2)
    .. controls (4.8,1.4) and (4.9,-0.2) .. (4.1,-1.1)
    .. controls (3.8,-2.1) and (2.9,-2.9) .. (1.9,-2.6)
    .. controls (0.9,-3.1) and (-0.5,-3.0) .. (-1.3,-2.4)
    .. controls (-2.3,-2.8) and (-3.2,-2.4) .. cycle;
  \node[fill=white,inner sep=1pt] at (3.9,-2.3) {$\mathcal S$};

  \fill (-0.9,0.0) circle (1.6pt);
  \node[below left,fill=white,inner sep=1pt] at (-0.9,0.0) {$s_0=c_0$};

  \draw[thick] (-0.9,0.0) -- (-0.15,0.32) -- (0.85,-0.10) -- (1.9,0.42) -- (3.0,-0.08) -- (4.0,0.55);
  \fill (0.85,-0.10) circle (1.4pt);
  \node[above,fill=white,inner sep=1pt] at (0.85,0.04) {$s_i$};
  \fill (3.0,-0.08) circle (1.4pt);
  \node[above,fill=white,inner sep=1pt] at (2.78,0.10) {$s_j$};
  \fill (4.0,0.55) circle (1.8pt);
  \node[below right,fill=white,inner sep=1.5pt] at (4.12,0.47) {$c_{m_0}$};

  \draw[black,thick,densely dashed] (0.85,-0.10) -- (1.25,0.95) -- (1.75,0.45) -- (2.35,1.28);
  \fill (2.35,1.28) circle (1.4pt);
  
  \node[above left,fill=white,inner sep=1pt] at (2.30,1.32) {$s_{i^\prime}$};

  \draw[black,thick,densely dashed] (3.0,-0.08) -- (3.3,0.78) -- (3.65,0.32) -- (3.95,1.02);
  \fill (3.95,1.02) circle (1.4pt);
  \node[above right,fill=white,inner sep=1pt] at (3.98,1.06) {$s_{j^\prime}$};

  \draw[densely dashed,->] (4.0,0.55) .. controls (3.55,1.75) and (0.15,1.95) .. (-0.9,0.0);
  \node[fill=white,inner sep=1pt] at (2.35,1.82) {reset after $m_0$ GS steps};

  \draw[black,thick] (0.60,-3.35) -- (1.30,-3.35);
  \node[right,fill=white,inner sep=1pt] at (1.35,-3.35) {GS};
  \draw[black,thick,densely dashed] (2.10,-3.35) -- (2.80,-3.35);
  \node[right,fill=white,inner sep=1pt] at (2.85,-3.35) {CS};

  \draw[thick] (-0.9,0.0) -- (-0.25,-0.42) -- (0.55,-0.08) -- (1.45,-0.52) -- (2.35,-0.18);
  \fill (2.35,-0.18) circle (1.6pt);
  \node[below right,fill=white,inner sep=1pt] at (2.35,-0.18) {$c_{2m_0}$};

  \draw[thick] (-0.9,0.0) -- (-0.2,-1.18) -- (0.65,-0.86) -- (1.55,-1.30) -- (2.55,-0.98);
  \fill (2.55,-0.98) circle (1.6pt);
  \node[below right,fill=white,inner sep=1pt] at (2.55,-0.98) {$c_{3m_0}$};

  \node[below left,fill=white,inner sep=1pt] at (2.43,-1.03) {$s_k$};
  \draw[black,thick,densely dashed] (2.55,-0.98) -- (2.52,-1.52) -- (2.00,-1.72) -- (1.45,-1.46) -- (1.00,-1.62);
  \fill (1.00,-1.62) circle (1.4pt);
  \node[below left,fill=white,inner sep=1pt] at (0.98,-1.66) {$s_{k^\prime}$};

\end{tikzpicture}
\caption{Schematic of the combination strategy: a main zigzag path from $s_0$
to $s_m$, secondary zigzag refinements from intermediate incumbents $s_i$, $%
s_j$ and $s_k$, and periodic GS restarts that regenerate main zigzags from $%
s_0=c_0$ after $m_0$ steps. Line style legend: solid lines for GS, dashed
lines for CS.}
\label{fig:combination_zigzags}
\end{figure}

In the illustrated example of Figure~\ref{fig:combination_zigzags}, along
the algorithm progression, the successive minima are denoted by $s_i,
s_{i+1}, \ldots, s_{i^\prime}$, $s_j, s_{j+1}, \ldots, s_{j^\prime}$, and $%
s_k, s_{k+1}, \ldots, s_{k^\prime}$.

\begin{theorem}
Under the assumptions of Proposition~\ref{prop:cs_converges}, the
algorithm reaches a global optimizer $s^*$ with probability $1$.
\end{theorem}

\begin{proof}
Let $B:=B_{\varepsilon}(s^*)$, where $\varepsilon$ is given by Proposition~\ref{prop:cs_converges}.
By Proposition~\ref{prop:parcours}, there exist an integer $m_0$ and a constant $p>0$ such that,
starting from $c_0$, one GS block of length $m_0$ reaches $B$ with strictly positive probability (say $p$).

In the combined strategy, the GS center is reset to $c_0$ after each block of $m_0$ GS steps.
Hence, at each block, the conditional probability of not hitting $B$ is at most $1-p$.
After $N$ blocks, we get
\[
\mathbb{P}(\text{never reaching }B\text{ during the first }N\text{ blocks}) \le (1-p)^N.
\]
Letting $N\to\infty$, we deduce
\[
\mathbb{P}(\tau_B<\infty)=1,
\]
where $\tau_B$ is the first hitting time of $B$.

Conditionally on the event $\{\tau_B<\infty\}$, we are in the setting of
Proposition~\ref{prop:cs_converges}; therefore the CS dynamics converges to $s^*$.
Since $\mathbb{P}(\tau_B<\infty)=1$, convergence to $s^*$ holds with probability $1$.
\end{proof}

\begin{remark}[Relation to classical random-search convergence theory]
\label{rem:solis_wets}
The above result is in the spirit of the classical random-search convergence
theory of \cite{solis1981minimization}, whose generic theorem however does
not apply off-the-shelf here: 

(i) its exploration condition requires that no
region of positive measure be persistently neglected, whereas our control
filter~\eqref{fo:alphamk} deterministically resets large coordinate moves and
could in principle block transitions (all the more as the scales
$\sigma_m^k=|c_m^k|$ degenerate near zero coordinates).

(ii) Proposition~\ref{prop:parcours} shows that reachability nevertheless holds
despite the control, with the explicit horizon bound~\eqref{eq:m0_bound} for
which the generic theory offers no analogue; the high-dimensional analysis of
the local regime (Lemma~\ref{lem:voisinage},
Proposition~\ref{prop:ring_trap}) is likewise outside its scope.
\end{remark}

\begin{remark}[Why periodic GS reset is needed]
If GS is not periodically reset to $c_0$, its center may drift and keep
exploring regions that do not lead to the target neighborhood, so the ball $%
B_{\varepsilon}(s^*)$ may not be reached in practice. Likewise, if GS is
stopped too early because the current incumbent seems ``good enough'', the
algorithm may remain trapped near a local minimum and may fall into the
annulus-trapping situation of Proposition~\ref{prop:ring_trap}, where $s^*$
is not reachable by CS.
\end{remark}

\subsection{Remarks}

\begin{itemize}
\item Both centers $c_m$ (GS) and $s_m$ (CS) remain in $\mathcal{S}$
throughout the algorithm. In particular, $s_m$ is always the current
best-so-far feasible candidate.

\item The control constant $C$ in~\eqref{fo:alphamk} is a user parameter.
When $C>2$, the GS reachability property holds from any initial center $c_0$
(Proposition~\ref{prop:parcours}).

\item The Gaussian choice is convenient but not essential. Theoretical
arguments extend to other log-concave proposal distributions; in practice,
the implementation can use either Gaussian or Laplace proposals.

\item In the CS regime, with $n$ chosen as a multiple of $K$, the practical
setting $n=K$ already yields strong concentration around the local center
when $K$ is large.
\end{itemize}

\vspace{-0.6cm}


\subsection{Synthesis}

For clarity, we summarize the method through its three practical blocks: GS,
CS, and their combination. Algorithm~\ref{algo:description} gives the
corresponding implementation template.


\SetKwProg{OwnObj}{Objective:}{}{} \SetKwProg{OwnGiv}{Given:}{}{} %
\SetKwProg{OwnInp}{Input:}{}{} \SetKwProg{OwnOut}{Output:}{}{} %
\SetKwProg{OwnEnd}{end}{}{} \SetKwProg{OwnReturn}{return}{}{}

\begin{algorithm}\label{algo:description}

\OwnGiv{}{a search space (constraint set) $\mathcal{S} \subset \mathbb{R}^{K}$ in high dimension; \\
a continuous function $f: \mathcal{S} \to \mathbb{R}$ to be minimized.\\
}

\BlankLine

\OwnObj{}{to find an approximation of an optimal solution.
}

\BlankLine

\OwnInp{}{
\nl an \textsc{initial candidate solution} $s_{0}$,\\
\nl a positive integer $m_0$ for periodic GS reset.
}

\BlankLine

\OwnOut{}{
the best candidate solution $s_{appr}^{*}$ found during the search.
}

\BlankLine

\nl $s_{appr}^{*} \leftarrow s_{0}$;

\BlankLine

\nl $m \leftarrow 0$;

\BlankLine

\nl $c_{0} \leftarrow s_0$ 

\BlankLine

\nl \While{\textsc{stopping criterion} is not met}{
\nl\If{$m$ is even (Global Search step)}{
\nl simulate an auxiliary point $\alpha_{m}$ according to \textsc{GS regime};\\
\nl update $s_{m+1}$ and $c_{m+1}$ according to the GS rules;
}
\nl\Else{
\nl simulate an auxiliary point $\alpha_{m}$ according to \textsc{CS regime};\\
\nl update $s_{m+1}$ and $c_{m+1}$ according to the CS rules;
}
\nl \If{the number of completed GS steps is a multiple of $m_0$}{
\nl reset the GS center: $c_{m+1} \leftarrow c_0$;
}
\nl $s_{appr}^{*} \leftarrow s_{m+1}$;

\BlankLine

\nl $m \leftarrow m+1$;

}

\nl \OwnEnd{}{}

\BlankLine

\vspace{-0.2cm}

\nl \OwnReturn{$s_{appr}^{*}$ and $f(s_{appr}^{*})$;}{} 

\caption{Synthesis of the GS--CS algorithm with periodic GS reset.}

\end{algorithm}


\section{Improving the LASSO\label{sec:lasso}}

We first illustrate our method on a LASSO regression problem (cf.~\cite%
{Tibshirani_1996, Chen_2022}). The standard objective is 
\begin{equation}
\min_{\theta \in \mathbb{R}^{K}} \sum_{i=1}^{M} \Big(y_{i} - \sum_{k=1}^{K}
z_{i}^{k} \cdot \theta^{k} \Big)^{2} + \lambda \cdot || \theta ||_{1} ,
\label{eq:lasso1}
\end{equation}
where $y_i$ are observations, $z_i^k$ are predictors (with $z_i^1:=1$), and $%
\lambda\ge 0$ is the regularization level. Let $\widehat{\theta}$ be a LASSO
solution. Our goal is to improve sparsity by finding $\theta$ with similar
predictive quality but smaller $\ell_1$ norm.


\subsection{Rewriting the problem}

\label{subsec:rewriting}

We use the classical link between the regularization-type LASSO and the hard-constrained-type basis pursuit denoising (BPDN)
(see e.g.~\cite{bayesian_perspective_2nd, lorenz2013necessaryconditions, Donoho_2006a}): \vspace{-0.3cm} 
\begin{eqnarray}
&& \min_{\theta \in \mathcal{S}} \, || \theta ||_{1}  \label{eq:bpdn1} \\%
[-0.2cm]
&& \text{with} \quad \mathcal{S} := \Big\{ \theta \in \mathbb{R}^{K}: \,
\sum_{i=1}^{M} \Big(y_{i} - \sum_{k=1}^{K} z_{i}^{k} \cdot \theta^{k} \Big)%
^{2} \leq \varepsilon \Big\}  \notag
\end{eqnarray}

for a tolerance parameter $\varepsilon>0$. In practice, the mapping between $%
\lambda$ and $\varepsilon$ is not explicit, so the two formulations are not
directly interchangeable. This motivates solving \eqref{eq:bpdn1} directly,
initialized at $\widehat{\theta}$.

Here the function $f$ is continuous, $\mathcal{S}$ is closed and we can suppose there exists $r$ such that $B_r(0)\cap \mathcal{S}=\emptyset$ (if not we translate space).

Using Scikit-learn, we evaluate each candidate with the coefficient of
determination 
\begin{equation}
r(\theta) := 1 - \frac{RSSQ(\theta)}{ \sum_{i=1}^{M} \Big(y_{i} - \overline{y%
} \Big)^{2} }  \notag
\end{equation}
with residual sum of squares $RSSQ(\theta) := \sum_{i=1}^{M} \Big(y_{i} -
\sum_{k=1}^{K} z_{i}^{k} \cdot \theta^{k} \Big)^{2}$ and data mean $%
\overline{y} := \frac{1}{M} \cdot \sum_{i=1}^{M} y_{i}$. Accordingly, 
\begin{equation}
\mathcal{S} = \Big\{ \theta \in \mathbb{R}^{K}: \, r(\theta) \geq r_{0} %
\Big\} ,  \label{fo:rewrittenS}
\end{equation}
with $\varepsilon = (1-r_0) \sum_{i=1}^{M} (y_i-\overline{y})^2$ and a
target score $r_0\in(0,1)$ close to $1$. We therefore search for points in %
\eqref{fo:rewrittenS} with smaller $\ell_1$ norm than $\widehat{\theta}$.

\begin{remark}[Norm vs. performance in the Lasso problem]
\label{rem:versus} Reducing $||\theta||_1$ usually decreases $r(\theta)$, so
optima are expected near the boundary $\partial\mathcal{S}$. Since $%
r_0\approx r(\widehat{\theta})$, the starting point $\widehat{\theta}$ is
already close to this boundary. Hence this experiment is a local-improvement
task by construction.
\end{remark}


\subsection{Data set and implementation\label{subsec:lasso_data}}

We apply our hybrid GS/CS method to a synthetic problem of dimension $K=5001$
generated with Scikit-learn. The dataset consists of $M=10{,}000$ samples and
$5{,}000$ features with $3{,}000$ informative variables (\texttt{random\_state=42}).
The LASSO baseline is computed with regularization parameter $\lambda=0.1$:

\begin{verbatim}
X, y, _ = make_regression(n_samples=10000, n_features=5000,
                          n_informative=3000, random_state=42)
clf = Lasso(alpha=0.1)
clf.fit(X, y)
\end{verbatim}

\noindent The implementation is available at \texttt{https://github.com/p052/GS\_CS\_algo.git}.

\noindent For the baseline LASSO solution, we obtain $\|\widehat{\theta}\|_1
\approx 147{,}328$ and $r(\widehat{\theta}) \approx 0.999996$. We set $c_0:=%
\widehat{\theta}$ and enforce $r_0=0.999$. Here $m_0$ is infinite, equivalently we do not resume $GS$ to its starting point.


\subsection{Results}

\subsubsection*{GS/CS over $2{,}000{,}000$ steps}

Running GS/CS for $2{,}000{,}000$ steps reduces the norm to
$141{,}864$, a total reduction of $5{,}327$ ($3.6\%$) from the Lasso
baseline. All $22{,}353$ improvements are produced by CS steps; GS produces
no update in this near-boundary regime, consistent with the theoretical role
of CS as the local refinement mechanism. The convergence is gradual but
persistent, with no sign of stagnation at $2M$ steps.

\begin{figure}[ht]
\centering
\includegraphics[width=\linewidth]{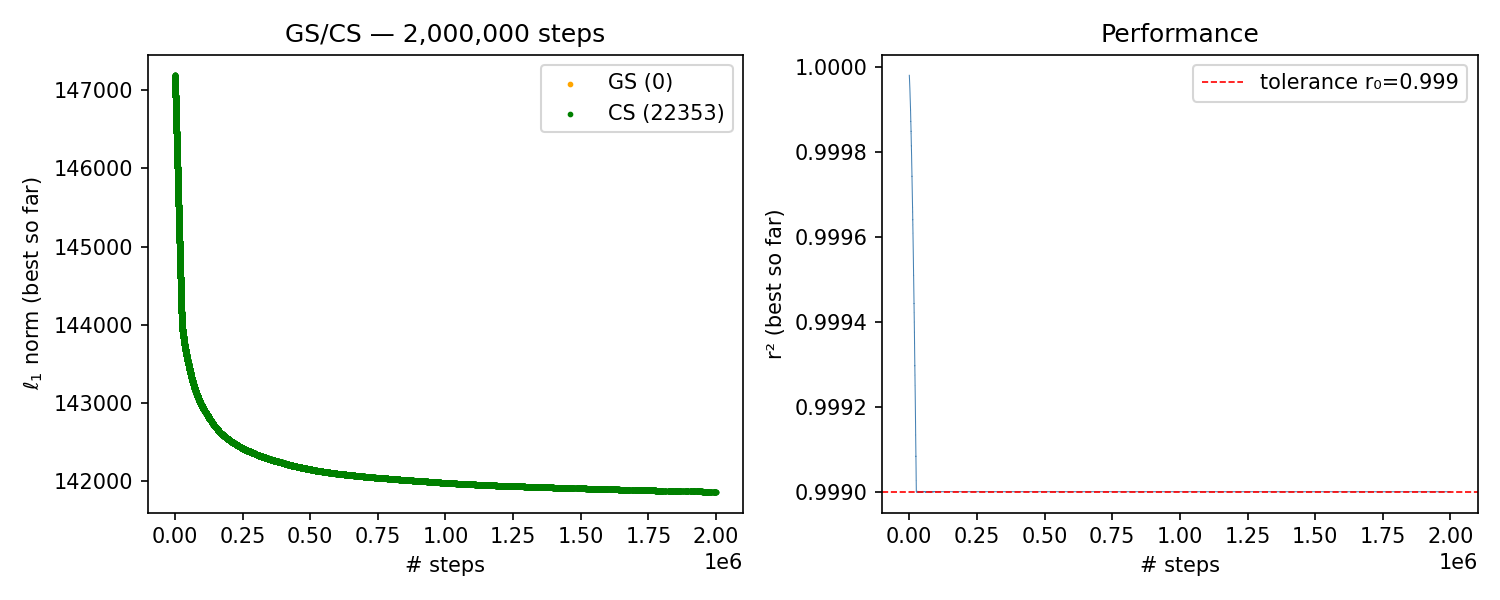}
\caption{Evolution of the $\ell_1$ norm (left) and performance (right)
over $2{,}000{,}000$ GS/CS steps; numbers in parentheses indicate the cumulative count of improvements at each marker.}
\label{fig:gscs_lasso_2M}
\end{figure}

\subsubsection*{Benchmark: Simulated Annealing\label{sssec:sa_lasso}}

As a first baseline, we ran simulated annealing (SA) on the
hard-constraint formulation~\eqref{eq:bpdn1}, minimizing the $\ell_1$ norm
under the feasibility constraint $r(\theta)\ge r_0$ (for SA in the
constrained continuous setting, see e.g.~\cite{romeijn1994simulated}). We used the
\texttt{dual\_annealing} routine of the SciPy library \cite{virtanen2020scipy}
in its pure-SA mode (\texttt{no\_local\_search=True}), started from the Lasso
solution $c_0=\widehat{\theta}$; the constraint was enforced strictly by
assigning the value $+\infty$ to any infeasible candidate, so that such
candidates are always rejected by the Metropolis rule.

Over $50{,}000$ function evaluations, SA produced \emph{zero improvements}:
the best feasible candidate remains the initial center $c_0$ itself
($\|\theta\|_1 = 147{,}328$). The main issue is
proposal scaling in high dimension: standard SA relies on isotropic
perturbations with one global variance, so all coordinates are perturbed at
the same scale. For heterogeneous high-dimensional parameters, some coordinates
require very small perturbations to maintain feasibility while others require
larger moves, making a single global scale inadequate. 

It is instructive to list the modifications that would be required to make SA
competitive on this problem: (i) per-coordinate proposal scales proportional
to the coordinate magnitudes $|\theta^k|$, instead of one global variance;
(ii) concentrated local moves around the incumbent, sized so that proposals
remain feasible near the thin boundary $r(\theta)\ge r_0$. These two
ingredients are precisely the elementwise adaptive scaling and the CS
concentration mechanism of the proposed method: a version of SA repaired
along these lines essentially coincides with GS/CS, up to the
temperature-controlled acceptance rule.

\subsubsection*{Benchmark: CMA-ES with augmented Lagrangian (pycma, diagonal mode)}

As a second baseline we ran CMA-ES \cite{hansen1996adapting} on the hard-constraint formulation~\eqref{eq:bpdn1} using the \texttt{pycma} library \cite{hansen2019pycma} with a
diagonal covariance matrix and an augmented-Lagrangian constraint handler.

\medskip
\noindent\textbf{Parameter choices.}
\begin{itemize}
  \item \textbf{Initial mean $\mu_0 = \widehat{\theta} = c_0$.}
  \item \textbf{Global step size $\sigma_0 = 1$, per-coordinate scaling
    $s_k = |c_0^k|/K$.}
    The library's \texttt{cma\_stds} option sets an independent step size
    per coordinate, matching the CS perturbation scale $|s_m^k|/(\sqrt{nK})$ with $n=1$.
  \item \textbf{Diagonal covariance (\texttt{CMA\_diagonal=True}).}
    Restricts the covariance matrix to its diagonal, reducing the
    per-generation cost from $O(K^3)$ to $O(K)$ and making the algorithm
    tractable when $K = 5{,}000$.
  \item \textbf{Augmented-Lagrangian constraint (\texttt{ConstrainedFitnessAL}).}
    The \texttt{pycma} constraint wrapper automatically adjusts Lagrange
    multipliers at each generation, providing an informative signal to
    distinguish near-feasible from deeply-infeasible candidates.
  \item \textbf{Maximum evaluations: $500{,}000$.}
\end{itemize}

\noindent\textbf{Results.}
CMA-ES finds 84 improvements over the first $\approx 5{,}000$ evaluations,
reducing $\|\theta\|_1$ from $147{,}328$ to $145{,}610$ ($-1{,}718$,
$-1.2\%$), then stagnates completely.  The augmented-Lagrangian handler
is key to this early progress: it adapts the feasibility signal in real time,
allowing the optimizer to walk along the constraint boundary.  However, the
diagonal covariance cannot capture correlations between coordinates, and the
fixed covariance structure offers no mechanism to escape the plateau once the
optimizer has locally exhausted its search directions.

\begin{figure}[ht]
\centering
\includegraphics[scale=0.45]{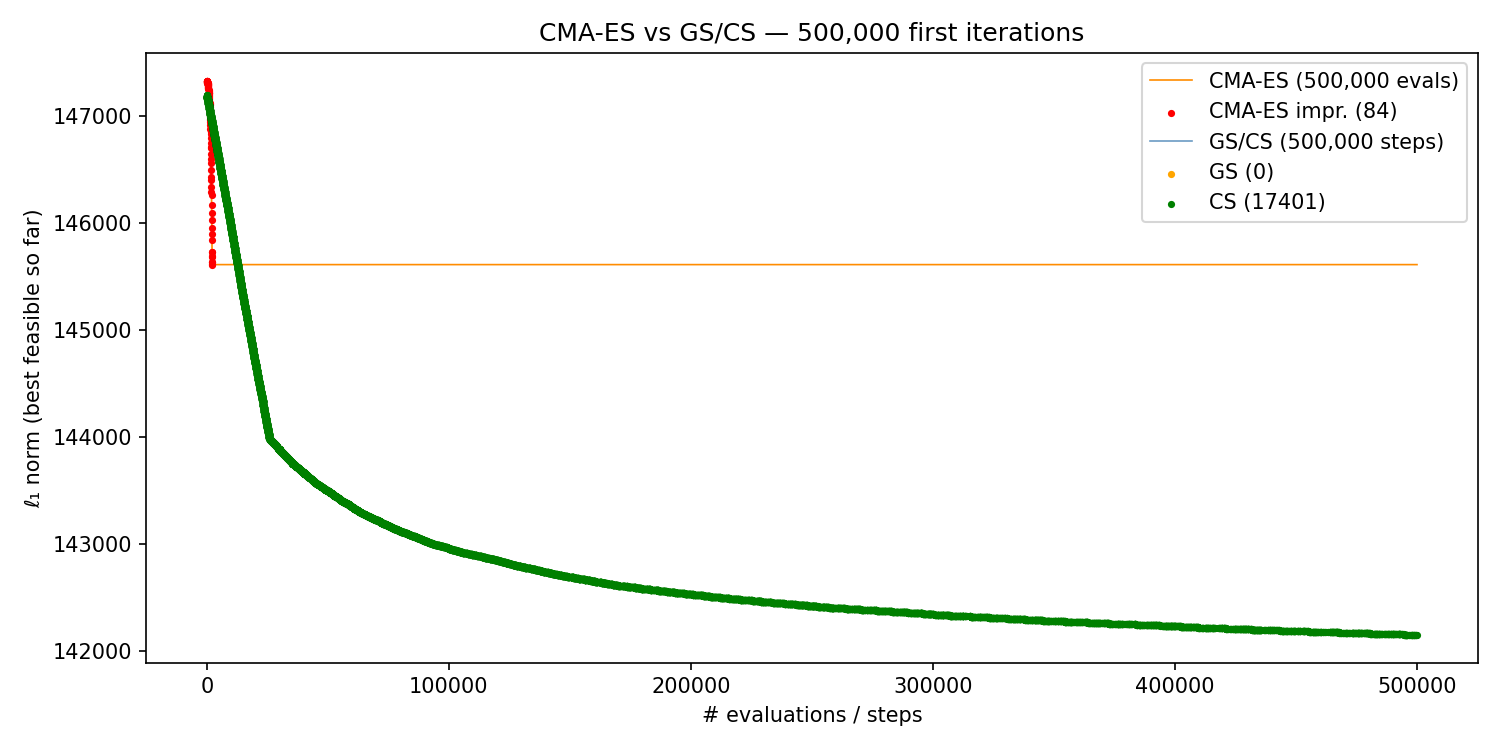}
\caption{Evolution of the best-so-far $\ell_1$ norm over $500{,}000$
iterations: \texttt{pycma} CMA-ES (red) vs.\ GS/CS (green), with improvement markers; numbers in parentheses indicate the cumulative count of improvements at each marker.}
\label{fig:cmaes_vs_gscs}
\end{figure}

\subsubsection*{Benchmark: LRA-CMA-ES — attempting to break the plateau}

The stagnation of \texttt{pycma} at $\|\theta\|_1 = 145{,}610$ is a classic symptom
of premature convergence. LRA-CMA-ES \cite{nomura2024lracmaes}, implemented
in the \texttt{cmaes} package \cite{nomura2024cmaes}, lowers the learning
rates when the gradient signal is weak and thereby maintains exploration far
longer than standard CMA-ES; in principle it could resume the decline that
\texttt{pycma} can no longer sustain. In practice, however, we could not make it
converge stably on this problem, for two compounding reasons tied to the high
dimensionality.

First, the \texttt{cmaes} package providing LRA-CMA-ES does not ship an
augmented-Lagrangian constraint handler, so the constraint $r(\theta)\ge r_0$
had to be enforced by a hand-coded static penalty
$\tilde f(a)=\|a\|_1+\rho\,\max(0,\,r_0-r(a))$; this produced no improvement
at any scale of $\rho$, since a fixed coefficient cannot track the thin
feasible ridge the way the augmented Lagrangian available in \texttt{pycma} does --- a
limitation of the constraint-handling interface rather than of LRA-CMA-ES
itself. Second, matching
\texttt{pycma}'s per-coordinate initialization instead requires a full initial
covariance, which the $2{,}000$ zero-valued BPDN~\eqref{eq:bpdn1}
coefficients render numerically singular (condition number $\approx 10^{16}$,
beyond float64 precision), so that the Cholesky update diverges within a few
generations. Combined with the $O(K^2)$ memory and arithmetic cost of
full-covariance CMA-ES in $K=5{,}000$ dimensions, these obstacles stem from
the dimensionality rather than from any intrinsic weakness of LRA-CMA-ES; a
separable (diagonal-covariance) variant would resolve both, at the price of
discarding covariance information.


\section{Improving Pruning Efficiency\label{sec:pruning}}

As another illustration of the relevance of the proposed method, we consider
neural-network compression by pruning.


\subsection{Problem specification}

\label{subsec:probspec}

We optimize pruning under a minimum-performance constraint. For weights $w\in%
\mathbb{R}^K$ and threshold $t>0$, define the pruning rate 
\begin{equation}
T_{t}(w) := \frac{1}{K} \cdot \sum_{k=1}^{K}{\mathbf{1}_{|w^{k}| < t}} \ ;
\label{fo:pruningrate}
\end{equation}
so $T_0(w)=0$ means no pruning. Let $P_t(w)$ be the classification
performance after pruning at threshold $t$. Given $t\geq 0$ and a required score 
$\eta\in(0,1)$ (typically $\eta>0.85$), we solve 
\begin{eqnarray}
&& \max_{w\in\mathcal{S}} T_{t}(w)  \label{eq:maxpruningrate} \\[-0.2cm]
&& \text{with} \quad \mathcal{S} := \mathcal{S}_{\eta,t} := \Big\{ w\in 
\mathbb{R}^K: \, P_{t}(w) > \eta \Big\} \, ,  \notag
\end{eqnarray}
that is, maximize pruning while preserving performance above $\eta$.
Equivalently, this is our minimization framework with $f(w):=-T_t(w)$. Here note that $f$ is no longer continuous;\footnote{Continuity can be restored at no practical cost. For instance, replacing each indicator $\mathbf{1}_{|w^{k}|<t}$ in~\eqref{fo:pruningrate} by the continuous summand $\mathbf{1}_{[0,\,t-\epsilon]}(|w^{k}|) + \frac{t-|w^{k}|}{\epsilon}\,\mathbf{1}_{[t-\epsilon,\,t]}(|w^{k}|)$ makes $T_t$, and hence $f$, continuous; taking $\epsilon$ smaller than the machine precision, this surrogate is numerically indistinguishable from~\eqref{fo:pruningrate} in all practical computations.} we nevertheless apply our algorithm without the above properties. Regarding the assumption $B_r(0)\cap \mathcal{S}=\emptyset$, having $w=0$ would lead to a neural network with a constant output and high performance would require the underlying classification problem to have only one class. Eventually $\mathrm{Cl}(\mathrm{Int}(\mathcal{S})) = \mathrm{Cl}(\mathcal{S})$ since $P_t$ is piece-wise continuous.


\subsection{Data set and implementation}

\label{subsec:cnndata}

We evaluate the method on Fashion-MNIST. Code for Section~\ref{sec:pruning}
is available at \texttt{https://github.com/p052/GS\_CS\_algo.git}.

\vspace{0.2cm} \noindent We first train the network shown in Figure~\ref%
{fig:nn_mnist}.

\begin{figure}[ht]
\begin{center}
\includegraphics[scale=0.5]{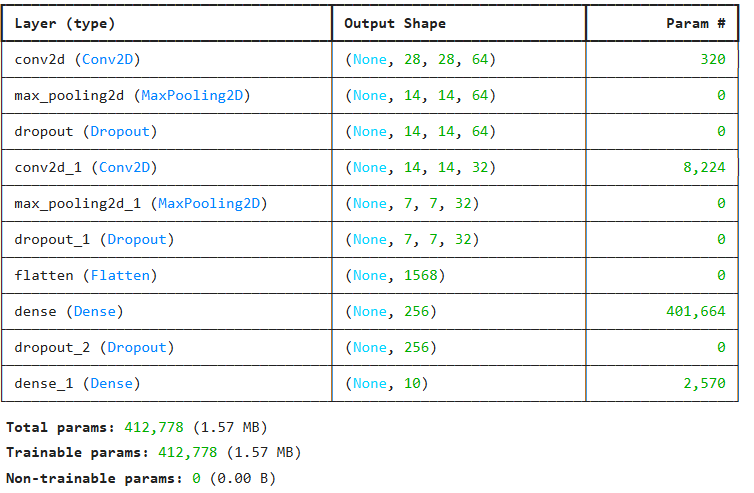}
\end{center}
\caption{Network architecture}
\label{fig:nn_mnist}
\end{figure}

\noindent The convolutional part has $8544$ weights ($w_{con}$), while the
dense part has $404234$ weights ($w$). All hidden activations are ReLU and
the final layer is softmax.

\vspace{0.2cm} \noindent After training, we obtain $\widehat{w}_{tot}:=(%
\widehat{w}_{con},\widehat{w})$ with $P_0(\widehat{w}_{tot})=0.91$ on test
data. We then freeze $\widehat{w}_{con}$ and optimize only the dense part $w$%
, starting from $s_0:=c_0:=\widehat{w}$ (dimension $K=404234$):

\begin{eqnarray}
&& \min_{w\in\mathcal{S}} \, (- T_{t}(w))  \label{eq:minpruningrate} \\%
[-0.2cm]
&& \text{with} \quad \mathcal{S} := \mathcal{S}_{\eta_{0},t} := \Big\{ w\in 
\mathbb{R}^K: \, P_{t}((\widehat{w}_{con},w)) > \eta \Big\} \, .  \notag
\end{eqnarray}
The feasibility condition $s_0\in\mathcal{S}$ is satisfied for small enough $%
t$ and $\eta<P_0((\widehat{w}_{con},\widehat{w}))$. In our experiments, we
set $t=0.01$ and $\eta=0.87$.

\vspace{0.2cm} \noindent We then run the GS/CS method (Section \ref{sec:algo}%
) with a permanent alternation strategy (Section~\ref{ssec:results}). This example also illustrates that the method works
with implicitly defined constraints, since feasibility is checked by
evaluating performance at each step. As in Remark~\ref{rem:versus}, pruning
can increase as long as performance remains strictly above the constraint
threshold.

\begin{remark}
\label{rem:convexity}[non-convexity of $\mathcal{S}$] Layer interactions
make $\mathcal{S}$ non-convex: even if $w_1,w_2\in\mathcal{S}$, $\frac{%
w_1+w_2}{2}$ may violate the performance constraint. For instance, we
obtained by training twice (empirically but reproducibly, see the
above-mentioned code link) two different weights vectors with performance
scores $P_{0}((\widehat{w}_{con},w_{1})) \approx 0.90 > 0.87 = \eta$ and $%
P_{0}((\widehat{w}_{con},w_{2})) \approx 0.90 > 0.87 = \eta$ such that $%
P_{0}((\widehat{w}_{con},\frac{w_{1}+ w_{2}}{2})) \approx 0.70 < 0.87 = \eta$%
. Also, $T_t$ is not concave (equivalently, $f=-T_t$ is not convex). If $%
w_1,w_2$ differ only in one coordinate with $w_1^1=2t+1$ and $w_2^1=0$, then 
\begin{equation*}
T_t(\frac{w_1+w_2}{2}) = T_t(w_1) < \frac{T_t(w_1) + T_t(w_2)}{2}
\end{equation*}
since $T_t(w_2) = T_t(w_1) + \frac{1}{K}$.
\end{remark}

\begin{remark}[possible non-connectivity of $\mathcal{S}$]
Generally, beyond non-convexity, $\mathcal{S}$ is not assumed to be connected either
(although path-connectedness is required in Proposition~\ref{prop:parcours}%
). Whether the particular feasible set induced by the above performance constraints is
connected seems to us an interesting research question. Indeed, the
diversity of weight configurations leading to a given level of performance
may make connectedness highly nontrivial. For instance, one region $\omega_1$
could rely mainly on the first half of the data coordinates, while another
region $\omega_2$ could rely on the second half; moving continuously from
one to the other may be difficult, but we do not prove this here.
\end{remark}


\subsection{Results\label{ssec:results}}

\subsubsection*{GS/CS results}

Recall the setup ($t=0.01$, $\eta=0.87$, $s_0=c_0=\widehat{w}$); we
test permanent alternation: GS at even steps, CS at odd steps. This avoids
tuning an explicit switching time while keeping both regimes active. Here again $m_0$ is infinite, equivalently we do not resume $GS$ to its starting point.

\vspace{0.2cm} \noindent Figure~\ref{fig:alternance_GS_CS} shows the evolution of the pruning rate and performance. Quantitatively, permanent alternation
reaches $35.7\%$ pruning at the tolerance $87\%$ performance.

\begin{figure}[ht]
\begin{minipage}[c]{.49\linewidth}
        \centering
        \includegraphics[width=\linewidth]{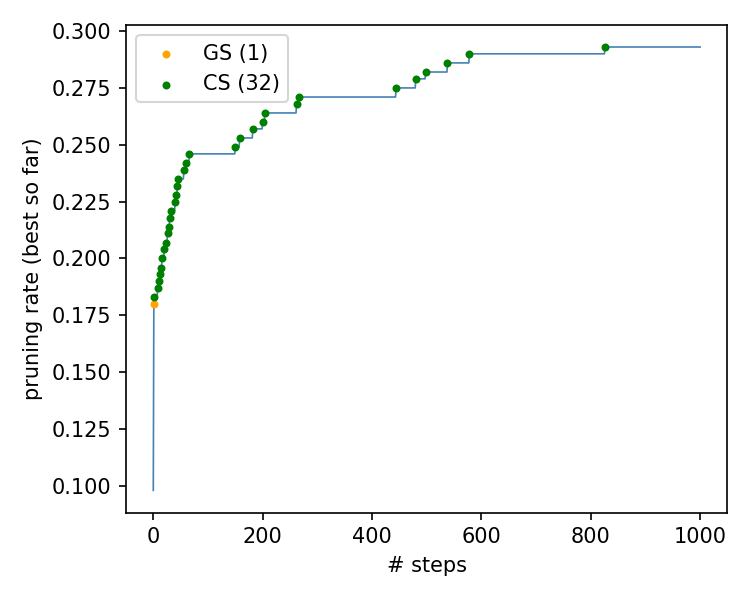}
    \end{minipage}
\hfill
\begin{minipage}[c]{.49\linewidth}
        \centering
        \includegraphics[width=\linewidth]{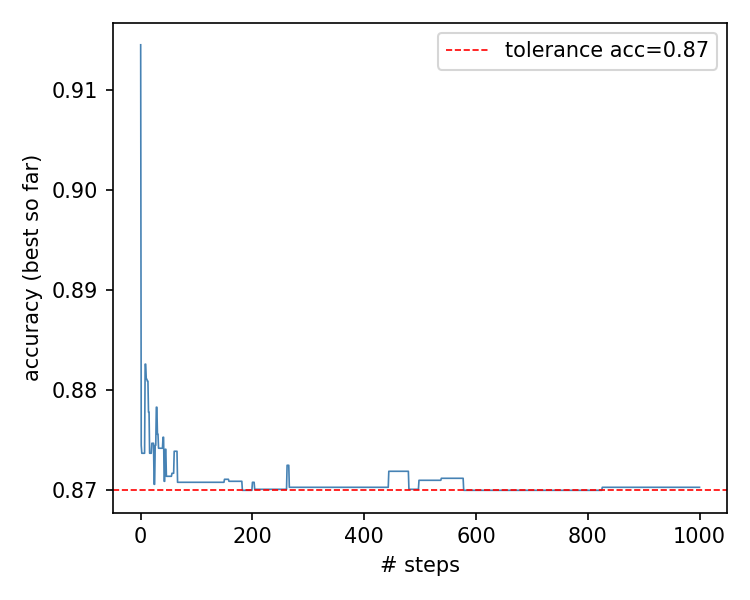}
    \end{minipage}
\caption{Evolution of $T_{t}(\cdot)$ (on the left) and $P_{t}((\protect%
\widehat{w}_{con},\cdot))$ (on the right) over $1000$ iterations, in the
case of permanent alternation}
\label{fig:alternance_GS_CS}
\end{figure}

\vspace{0.2cm} \noindent Overall, permanent alternation is competitive with
explicit switching, while removing the need to detect proximity to the
boundary.

\subsubsection*{Benchmark: Simulated Annealing}

As in Subsection~\ref{sssec:sa_lasso}, we ran a constrained simulated annealing (SA)
baseline for problem~\eqref{eq:minpruningrate}, maximizing the pruning rate
$T_t(\cdot)$ under the performance constraint $P_t((\widehat{w}_{con},\cdot))>\eta$.
SA did not manage to improve the initial center $c_0$, for the same reason as
in the BPDN setting: isotropic single-scale proposals cannot match the
heterogeneous coordinate magnitudes of neural-network weights, whereas GS/CS
perturbs each coordinate $k$ at a scale proportional to its current magnitude
($|c_m^k|$ in GS and $|s_m^k|/(\sqrt{nK})$ in CS).

\subsubsection*{Benchmark: CMA-ES}

The pruneable weight vector has dimension $K = 404{,}330$ (dense layers and
biases only; convolutional kernels are kept fixed). Standard CMA-ES maintains
a full $K \times K$ covariance matrix, requiring approximately
$K^2 \times 8\,\text{bytes} \approx 1{,}308\,\text{GB}$ of memory — far
beyond the capacity of any standard machine. Even the diagonal-covariance
variant, which reduces storage to $O(K)$, requires a population of
$\lambda = 4 + \lfloor 3\ln K \rfloor = 42$ candidates per generation, each
of dimension $K$, amounting to roughly $136\,\text{MB}$ per generation.
More critically, CMA-ES adapts a global covariance structure that is
fundamentally ill-suited to the heterogeneous coordinate scales of neural
network weights, where magnitudes vary by several orders of magnitude across
layers. For these reasons, CMA-ES is not a practical baseline for this
problem, and we restrict the comparison to the LASSO setting where $K=5{,}001$.

\section{Conclusion}

We proposed a hybrid stochastic optimization method for high-dimensional
constrained continuous global problems without convexity assumptions on either the objective or
the feasible set. The method combines Global Search (GS), for broad
exploration, and Concentrated Search (CS), for local refinement via
concentration effects.

A central ingredient is the elementwise definition of proposal scales in both
regimes. This yields non-isotropic and adaptive perturbations, aligned with
heterogeneous coordinate magnitudes, while avoiding coordinate-specific
hyperparameter tuning.

Our theoretical analysis highlights complementary strengths: GS reaches any
target region of $\mathcal{S}$ in finite time with positive probability (under
path-connectedness assumptions), while CS provides effective local behavior,
particularly near constraint boundaries.

The two applications support this picture. In the BPDN setting, the method
produces a consistent reduction of the $\ell_1$ norm ($-3.6\%$ over
$2{,}000{,}000$ steps) under a strict performance constraint, outperforming
CMA-ES which stagnates after fewer than $5{,}000$ evaluations. In
neural-network pruning on Fashion-MNIST, permanent alternation reaches
$35.7\%$ pruning while preserving accuracy above the tolerance, in a regime
where CMA-ES is memory-infeasible.

Overall, the approach provides a practical and theoretically grounded
framework for non-convex high-dimensional constrained continuous global optimization,
whose elementwise scaling addresses the main limitation observed with the
isotropic baselines in our experiments.
Future work includes extensions to other architectures, larger-scale tasks,
and combinations with additional compression methods.


\section*{Acknowledgements}

W. Stummer gratefully acknowledges partial financial support from Sorbonne
Universit\'{e} Paris and the hospitality of LPSM. M. Broniatowski gratefully
acknowledges partial financial support from FAU Erlangen-N{\"u}rnberg and
its hospitality.

\newpage
\section*{Declarations}

\noindent\textbf{Funding.} See the title-page footnote and the Acknowledgements above.
\textbf{Conflict of interest.} The authors declare that they have no
conflict of interest.
\textbf{Data availability.} The datasets used in this article are
synthetic (Scikit-learn \texttt{make\_regression}) or publicly available
(Fashion-MNIST). The full implementation and experiment code are available at\\
\texttt{https://github.com/p052/GS\_CS\_algo.git}.


\bibliographystyle{plain}
\bibliography{Biblio}

\appendix

\section{Appendix}

\subsection{Proof of Proposition~\protect\ref{prop:parcours}}

\label{app:proof_prop_parcours} Since $\mathcal{S}$ is path-connected, there
exists a path $\gamma : [0,1] \to \mathcal{S}$ such that $\gamma(0) = c_0$
and $\gamma(1) = \beta$. We will construct a sequence of centers $c_0,
\ldots, c_{m_0}$ under the GS regime such that the probability of
transitioning from one to the next is nonzero, and such that the target ball 
$B_{\eta}(\beta)$ becomes reachable from $c_{m_0}$.

During the GS regime, any point can be reached in a single step with nonzero
probability (albeit potentially small) by $a_m$. However, we have imposed a
control mechanism (namely,~\eqref{fo:alphamk}) for transitioning from $a_m$
to $\alpha_m$. We will show that despite this control, the ball around $\beta
$ remains reachable in a finite number of iterations.

The path begins at $c_0$, an arbitrary point in $\mathcal{S}$. First, we
define $\underline{\delta_0} = \min_{k} |c_0^k|$
and, by the proposition hypothesis, we have $\underline{\delta_0}\geq r>0$.
We then assume without loss of generality that $\eta < \underline{\delta_0}$.

\noindent For $\delta \in \mathbb{R}_{+}^K$ and $c \in \mathbb{R}^K$, we
define the hyperrectangle
\begin{equation*}
H_{C*\delta}(c) = \{x \in \mathbb{R}^K ~|~ |x^k - c^k| \leq C * \delta^k
\text{ for all k}\} ;
\end{equation*}
additionally, we define $H_m = H_{C*\delta_m}(c_m) - int\left(H_{2*\delta_m}(c_m)\right)$,
where $\delta_m^k = |c_m^k|$ for all $1\le k \le K$.
Since $C > 2$, $H_m$ is nonempty. The purpose of this ``annulus'' $H_m$ is
to maintain $\min_k \delta_m^k \geq \underline{\delta_0}$ along the path $%
c_0,\ldots,c_{m_{0}}$; indeed, along iterations, the absolute value of each
coordinate is strictly increasing. This ensures that $\delta_{m+1}^k >
\delta_m^k > \underline{\delta_0}$ for all $k$. This allows~%
\eqref{fo:alphamk} to be less and less restrictive and therefore increases
the speed along path $\gamma$. Let us now restrict realizations on a tubular
neighborhood of $\gamma$ included in $\mathcal{S}$. We define
$t_m = \sup \{t \in [0,1] ~|~ \gamma(t) \in H_{C*\delta_m}(c_m)\}$
and $\beta_m = \gamma(t_m)$.

\textit{Case 1}: $t_m = 1$. In this case, $\beta_m = \beta \in
H_{C*\delta_m}(c_m)$, and $B_{\eta}(\beta)$ becomes reachable by $\alpha_m$
with nonzero probability. Specifically, "$a_m \in \left(B_{\eta}(\beta) \cap
H_{C*\delta_m}(c_m)\right)$" is an event with nonzero probability.
Furthermore, $a_m \in H_{C*\delta_m}(c_m)$ implies that the control~%
\eqref{fo:alphamk} is satisfied and $\alpha_m = a_m$, which concludes the
argument: $m=m_0$.

\textit{Case 2}: $t_m < 1$. In this case, there exists $\epsilon \in [0,
\eta]$ such that $B_{\epsilon}(\beta_m) \subset \mathcal{S}$ because $%
\mathcal{S}$ is open. Now, "$a_m \in \left(B_{\epsilon}(\beta_m) \cap
H_m\right)$" is an event with nonzero probability. Moreover, $a_m \in H_m$
implies that the control~\eqref{fo:alphamk} is satisfied and $\alpha_m = a_m$%
. Since $B_{\epsilon}(\beta_m) \subset \mathcal{S}$, we have $\alpha_m \in
\mathcal{S}$, and thus $c_{m+1} = \alpha_m = a_m$.
Additionally, by the definition of $H_m$, we have that $\delta_{m+1}^k >
\delta_m^k > \underline{\delta_0}$ for all $k$. Next, we analyze the path
between $c_m$ and $c_{m+1}$. Since $c_{m+1} \in B_{\epsilon}(\beta_m)$ and $%
\epsilon < \eta < \underline{\delta_0} \leq \min_k \delta_{m+1}^k$, we have $%
\beta_m \in B_{\epsilon}(c_{m+1}) \subset H_{C*\delta_{m+1}}(c_{m+1})$.
Furthermore, because of $C > 2$ we obtain 
\begin{equation*}
B_{\underline{\delta_0}}(\beta_m) \subset B_{\underline{\delta_0}%
+\epsilon}(c_{m+1}) \subset H_{2*\delta_{m+1}}(c_{m+1}) \subset
H_{C*\delta_{m+1}}(c_{m+1}) . 
\end{equation*}
Thus, $\beta_m\in H_{C*\delta_{m+1}}(c_{m+1})$, leading to $t_{m+1} \geq t_m$%
. Now, either $t_{m+1} = 1$ or $||\beta_{m+1} - \beta_m||_1 \geq \underline{%
\delta_0}$; indeed, $\beta_{m+1}$ belongs to the annulus centered at $c_{m+1}
$ and with radius greater than $\underline{\delta_0}$.

In conclusion, we construct a sequence of points
$\beta_0 = \gamma(0) = c_0, \ldots, \beta_m = \gamma(t_m), \ldots,
\beta_{m_0+1} = \gamma(1) = \beta$
along the path $\gamma.$ To these points are associated a list of centers $%
c_0, \ldots, c_{m_0}$ reached by the algorithm under the GS regime and which
respect $c_m\in B_{\eta}(\beta_m)$ for all $m$.

The finite integer $m_0$ exists because at each step, $t_{m+1} > t_m$ and $%
||\beta_{m+1} - \beta_m||_1 \geq \underline{\delta_0}$ or $t_{m+1} = 1$.
Moreover, 
\begin{equation}
m_0 \leq \frac{||\gamma||_1}{\underline{\delta_0}}  \label{eq:m0_bound}.
\end{equation}

\subsection{Proof of Lemma~\protect\ref{lem:voisinage}}

\label{app:proof_lem_voisinage} Let's begin by making the random vector
isotropic. We observe that 
\begin{eqnarray*}
&&\mathbb{P}\left( ||\alpha_{m,n}(\sigma)-s_m||_{2}\notin [(1-\theta )S_{n}%
\sqrt{K},(1+\theta )S_{n}\sqrt{K}]\right) \\
&\leq &\mathbb{P}\left(||\alpha_{m,n}(V_{\bar{\sigma}})-s_m||_{2}\notin [
(1-\theta )S_{n}\sqrt{K},(1+\theta )S_{n}\sqrt{K}]\right) .
\end{eqnarray*}
Now, $\frac{\alpha_{m,n}(V_{\bar{\sigma}})-s_m}{S_{n}}$ satisfies the
assumptions in~\cite{guedon:hal-00795791_new}, and therefore 
\begin{eqnarray*}
&&\mathbb{P}\left( \left\vert ||\alpha_{m,n}(V_{\bar{\sigma}})-s_m||_{2}-S_{n}%
\sqrt{K} \right\vert \geq S_{n}\sqrt{K}t\right)\\
&&=\mathbb{P}\left(
\left\vert~\left\vert\left\vert \frac{\alpha_{m,n}(V_{\bar{\sigma}})-s_m}{%
S_{n}}\right\vert\right\vert_{2}-\sqrt{K}~\right\vert \geq t \sqrt{K}\right)
 \leq C_{1}\exp \left( -C_{2}tK^{1/2}\right)
\end{eqnarray*}
for positive constants $C_{1}$ and $C_{2}$ independent of $n$ and all
positive $t$.

For a given $\eta\in(0,1)$, we define 
\begin{equation}
\theta_{\eta} = \frac{\ln(C_1)-\ln(1-\eta)}{C_2\sqrt{K}} .
\label{cond theta}
\end{equation}

As long as $K$ is large enough to have $\frac{\ln(C_1)}{C_2\sqrt{K}}< 1$,
there exists a $\eta_K$ which satisfies the claim, namely
$\eta_K = \sup \left\{\eta~/~\eta>0 ~\&~\theta_{\eta} < 1\right\}$;
furthermore, $\lim_{K\to\infty}\eta_K = 1$.

\end{document}